\documentclass[11pt,a4paper]{article}
\pdfoutput=1

\makeatletter
\providecommand{\@listi}{\leftmargin\leftmargini}
\makeatother

\usepackage{jheppub}

\usepackage{amsthm}
\usepackage{booktabs}
\usepackage{multirow}
\usepackage{array}
\usepackage{slashed}
\usepackage{siunitx}
\usepackage{xcolor}
\usepackage{tikz}
\usepackage{pgfplots}
\usepackage{float}

\usetikzlibrary{positioning,arrows.meta,decorations.markings,shapes}
\pgfplotsset{compat=1.17}

\newcommand{\GeV}{\,\mathrm{GeV}}

\newcommand{\eV}{\,\mathrm{eV}}
\newcommand{\vev}[1]{\langle #1 \rangle}

\newcommand{\eps}{\varepsilon}

\newcommand{\MR}{M_R}
\newcommand{\MD}{M_D}
\newcommand{\Mnu}{M_\nu}

\newtheorem{theorem}{Theorem}[section]
\newtheorem{lemma}[theorem]{Lemma}

\title{From seesaw over-suppression to trimaximal mixing:\\
why \boldmath $A_4$ is the minimal resolution of the \boldmath $Z_3$ neutrino failure}

\author{Navid Ardakanian}
\affiliation{Independent Researcher}
\emailAdd{n.ardakanian@gmail.com}

\abstract{We investigate whether the type-I seesaw mechanism can rescue the $Z_3$
Froggatt--Nielsen framework for neutrinos and find that it cannot.  With
right-handed Majorana masses carrying the $Z_3$ charge structure dictated
by the Majorana bilinear---where suppression powers follow
$(q_i+q_j)\bmod 3$---the mass matrix contains an unsuppressed
off-diagonal entry whose dominance in $M_R^{-1}$, combined with the
hierarchical column texture of $M_D$, over-suppresses the two lightest
neutrino masses to $\mathcal{O}(\varepsilon^3)$ while $m_3$ remains
$\mathcal{O}(1)$.  This pushes the solar-to-atmospheric mass ratio to a
median $\Delta m^2_{21}/\Delta m^2_{31}\sim 4\times 10^{-11}$---eight
orders of magnitude below the observed value of $0.030$.  We prove this
failure is universal across all six permutations of the charges
$(2,1,0)$ and show analytically that the generic ratio scales as
$\Delta m^2_{21}/\Delta m^2_{31}\sim\mathcal{O}(\varepsilon^6)
\sim 10^{-11}$, with fewer than $0.01\%$ of parameter-space points
exceeding $\varepsilon^2\approx 2\times 10^{-4}$.  The PMNS angles
remain Haar-random, carrying no information from the expansion parameter.
We then show that $A_4$, the alternating group of order 12, is the
minimal discrete symmetry resolving both failures.  Its triplet
representation provides two independent vacuum parameters controlling
the solar and atmospheric mass scales separately, while constraining the
PMNS matrix to the trimaximal TM$_1$ pattern.  The TM$_1$ solar sum
rule predicts $\sin^2\theta_{12}=0.318$ ($1.2\sigma$ from NuFit 6.0,
$1.0\sigma$ from JUNO), and the atmospheric sum rule yields a
parameter-free $(\sin^2\theta_{23},\,\cos\delta)$ correlation predicting
$\delta\approx -71^\circ$, testable at DUNE and T2HK.}


\arxivnumber{2603.21264}

\begin{document}
\maketitle
\flushbottom

%%%%%%%%%%%%%%%%%%%%%%%%%%%%%%%%%%%%%%%%%%%%%%%%%%%%%%%%%%%%
\section{Introduction}
\label{sec:intro}
%%%%%%%%%%%%%%%%%%%%%%%%%%%%%%%%%%%%%%%%%%%%%%%%%%%%%%%%%%%%

The flavor puzzle remains one of the deepest open questions in particle physics.
Fermion masses span twelve orders of magnitude---from sub-eV neutrinos to the
$173\GeV$ top quark---and the mixing patterns differ qualitatively between quarks
and leptons: the Cabibbo--Kobayashi--Maskawa (CKM) matrix exhibits small,
hierarchical mixing angles, while the Pontecorvo--Maki--Nakagawa--Sakata (PMNS)
matrix features two large angles and one small
one~\cite{Esteban:2024nav}. No principle within the Standard Model explains
either the mass hierarchies or the mixing dichotomy.

Discrete flavor symmetries offer a systematic approach to this
problem~\cite{Altarelli:2010gt,Ishimori:2010au,King:2013eh}. Among the simplest
frameworks is the Froggatt--Nielsen (FN) mechanism~\cite{Froggatt:1979np} based
on a $Z_3$ symmetry. As demonstrated in ref.~\cite{Ardakanian:Z3}, a single
expansion parameter $\eps \simeq 0.015$, fixed from the ratio $m_u/m_t$,
structurally accounts for the hierarchical pattern of quark and charged-lepton
mass ratios with natural $\mathcal{O}(1)$ Yukawa couplings. The model assigns
generation-dependent $Z_3$ charges only to the right-handed fermions, with
left-handed doublets uncharged, producing a \emph{column texture} in which
entry $(i,j)$ of the mass matrix carries $\eps^{q_j}$ regardless of $i$. This
column texture is the defining feature of the framework: it structurally predicts
the mass hierarchy $m_1:m_2:m_3 \sim \eps^2:\eps:1$ across the quark and
charged-lepton sectors~\cite{Ardakanian:Z3}.

However, the $Z_3$ framework fails for neutrinos on two fronts. First, the mass
spectrum is far too hierarchical: the column texture predicts
$\Delta m^2_{21}/\Delta m^2_{31} \lesssim 10^{-4}$, whereas the observed ratio
is $0.030$~\cite{Esteban:2024nav}. Second, the PMNS mixing angles carry no
information from $\eps$---they are generically $\mathcal{O}(1)$, consistent with
Haar-distributed random unitaries~\cite{deGouvea:2012ac}, providing no angular
structure whatsoever~\cite{Ardakanian:Z3}.

This raises a sharp question: \emph{can the type-I seesaw mechanism, combined
with a right-handed Majorana mass matrix $\MR$ carrying the $Z_3$ charge
structure dictated by the Majorana bilinear, rescue the framework for
neutrinos?}

In this paper---the second in a two-paper sequence following
ref.~\cite{Ardakanian:Z3}---we answer this question comprehensively, and then
identify the minimal discrete symmetry that resolves the failure.

We show that the correct $Z_3$ Majorana charge algebra, where $\eps$-powers are
determined by $(q_i + q_j) \bmod 3$ rather than $q_i + q_j$, produces a
qualitatively unexpected $\MR$ structure containing an unsuppressed off-diagonal
entry. This entry dominates $\MR^{-1}$, and when combined with the hierarchical
column texture of $\MD$, it over-suppresses both $m_1$ and $m_2$ to
$\mathcal{O}(\eps^3)$ while $m_3$ remains $\mathcal{O}(1)$, \emph{deepening}
the mass spectrum failure from $\sim 10^{-4}$ to $\sim 10^{-11}$. The median
predicted ratio $\Delta m^2_{21}/\Delta m^2_{31} \sim 4 \times 10^{-11}$---eight
orders of magnitude below the observed value---and only a fraction $\sim 10^{-5}$
of the $\mathcal{O}(1)$ parameter space comes within a factor of~3 of the data
through accidental cancellations.

We prove this failure is \emph{universal}: it holds identically for all six
permutations of the $Z_3$ charges $(2,1,0)$ among the three generations. The
proof rests on a number-theoretic property---the pair with charges 1 and 2 always
sums to $3 \equiv 0 \pmod{3}$, guaranteeing an unsuppressed off-diagonal entry
in every case. All six permutations yield the same median
$R \sim 4 \times 10^{-11}$.

Having established the precise structural failure of $Z_3$ and identified its
origin, we then ask: \emph{what is the minimal discrete symmetry that resolves
the two-scale problem?} In sections~\ref{sec:why_nonabelian}--\ref{sec:conclusions}
we show that $A_4$, the alternating
group of order~12, provides exactly the additional structure needed. Its triplet
representation introduces two independent vacuum parameters that control the two
neutrino mass scales independently, while simultaneously constraining the PMNS
mixing pattern.

The paper is organized as follows. Section~\ref{sec:framework} establishes the
$Z_3$ seesaw framework, including the column texture and the correct Majorana
charge algebra. Section~\ref{sec:inverse} presents the inverse seesaw
reconstruction, demonstrating that the required $\MR$ is consistent with $Z_3$
charge scaling. Section~\ref{sec:analytical} identifies the seesaw
over-suppression mechanism through forward scans and analytical derivation, showing that
the generic ratio scales as $\mathcal{O}(\eps^6)$.
Section~\ref{sec:universal} proves that the failure is universal across all
charge permutations. Section~\ref{sec:why_nonabelian} explains why non-abelian
structure is required. Section~\ref{sec:A4} develops the $A_4$ resolution.
Section~\ref{sec:comparison} compares TM$_1$ predictions with data.
Section~\ref{sec:discussion} discusses implications, and
section~\ref{sec:conclusions} concludes.

%%%%%%%%%%%%%%%%%%%%%%%%%%%%%%%%%%%%%%%%%%%%%%%%%%%%%%%%%%%%
\section{The \boldmath $Z_3$ seesaw framework}
\label{sec:framework}
%%%%%%%%%%%%%%%%%%%%%%%%%%%%%%%%%%%%%%%%%%%%%%%%%%%%%%%%%%%%

\subsection{The \boldmath $Z_3$ Froggatt--Nielsen column texture}
\label{sec:column_texture}

We extend the Standard Model by a complex scalar flavon $\Phi$ transforming as
$\Phi \to \omega\Phi$ under $Z_3$, where $\omega = e^{2\pi i/3}$. The SM Higgs
$H$ is $Z_3$-neutral. The minimal charge assignment is
\begin{align}
    Q_L^i,\; L_L^i &\sim 0 \quad (i=1,2,3), \nonumber\\
    f_R^1 &\sim 2, \quad f_R^2 \sim 1, \quad f_R^3 \sim 0,
    \label{eq:charges}
\end{align}
where $f_R$ denotes any right-handed fermion ($u_R, d_R, e_R, \nu_R$). Since
the left-handed doublets are $Z_3$-neutral, the $\eps$-suppression depends only
on the right-handed charge, yielding the column texture
\begin{equation}
    \MD^f = \frac{v_H}{\sqrt{2}}
    \begin{pmatrix}
        c_{11}^f\, \eps^2 & c_{12}^f\, \eps & c_{13}^f \\[2pt]
        c_{21}^f\, \eps^2 & c_{22}^f\, \eps & c_{23}^f \\[2pt]
        c_{31}^f\, \eps^2 & c_{32}^f\, \eps & c_{33}^f
    \end{pmatrix},
    \label{eq:MD}
\end{equation}
where every entry in column $j$ carries $\eps^{q_j}$ with $q_j = 2, 1, 0$, and
the $c_{ij}^f$ are $\mathcal{O}(1)$ complex Yukawa couplings. This structure
generates mass hierarchies $m_1:m_2:m_3 \sim \eps^2:\eps:1$ as a structural
prediction~\cite{Ardakanian:Z3}.

A crucial consequence of the column texture is that the left-handed unitary
rotations $U^f_L$ diagonalizing $\MD^f (\MD^f)^\dagger$ are determined by the
directions of the column vectors
$\vec{c}^f_j = (c^f_{1j}, c^f_{2j}, c^f_{3j})^T$ in flavor space. Since these
directions are set by random $\mathcal{O}(1)$ complex numbers, $U^f_L$ is
generically an $\mathcal{O}(1)$ unitary rotation. Both the CKM matrix
$V_{\text{CKM}} = U^{u\dagger}_L U^d_L$ and the PMNS matrix
$U_{\text{PMNS}} = U^{\ell\dagger}_L U^\nu_L$ are therefore generically random
unitaries, carrying no parametric information from
$\eps$~\cite{Ardakanian:Z3}.

\subsection{Type-I seesaw}
\label{sec:seesaw}

In the type-I seesaw, the light neutrino mass matrix is
\begin{equation}
    \Mnu = -\MD\, \MR^{-1}\, \MD^T,
    \label{eq:seesaw}
\end{equation}
where $\MD$ is the neutrino Dirac matrix with the column texture of
eq.~\eqref{eq:MD}.

\subsection{The \boldmath $Z_3$ charge structure of \boldmath $\MR$}
\label{sec:MR_charges}

The Majorana mass term $\nu_R^{cT} C \MR \nu_R$ is a bilinear, so the $Z_3$
charge of the $(i,j)$ entry is $q_i + q_j$. For $Z_3$-invariance, each entry
must be dressed with flavon insertions such that the total charge vanishes
modulo~3. The minimum number of insertions is
\begin{equation}
    n_{ij} = \begin{cases}
        0 & \text{if } (q_i + q_j) \equiv 0 \pmod{3}, \\[3pt]
        1 & \text{if } (q_i + q_j) \equiv 1 \text{ or } 2 \pmod{3}.
    \end{cases}
    \label{eq:power_rule}
\end{equation}
For entries with charge sum $\equiv 1\pmod{3}$, a single $\Phi$ insertion
(charge~$+1$) cancels the residual charge.  For entries with charge sum
$\equiv 2\pmod{3}$, a single $\Phi^\dagger$ insertion (carrying charge
$-1\equiv 2\pmod{3}$) suffices, with
$\vev{\Phi^\dagger}/\Lambda=\eps$.\footnote{Equivalently, one may
introduce a conjugate flavon $\bar{\Phi}$ with $Z_3$ charge~2 and
$\vev{\bar{\Phi}}/\Lambda=\eps$.  Both choices give $n_{ij}=1$ for all
nonzero residues.}  Note the crucial role of the modular arithmetic: the
$\eps$-power is determined by $(q_i + q_j) \bmod 3$, \emph{not} by
$q_i + q_j$ itself.

For the standard charge assignment $(q_1, q_2, q_3) = (2, 1, 0)$, we compute:
\begin{equation}
    (q_i + q_j) \bmod 3 =
    \begin{pmatrix}
        1 & 0 & 2 \\
        0 & 2 & 1 \\
        2 & 1 & 0
    \end{pmatrix},
    \quad\Longrightarrow\quad
    n_{ij} =
    \begin{pmatrix}
        1 & 0 & 1 \\
        0 & 1 & 1 \\
        1 & 1 & 0
    \end{pmatrix}.
    \label{eq:powers_210}
\end{equation}
The resulting Majorana mass matrix is
\begin{equation}
    \MR = M_0
    \begin{pmatrix}
        a_{11}\,\eps   & a_{12}         & a_{13}\,\eps \\[2pt]
        a_{12}         & a_{22}\,\eps   & a_{23}\,\eps \\[2pt]
        a_{13}\,\eps   & a_{23}\,\eps   & a_{33}
    \end{pmatrix},
    \label{eq:MR_structure}
\end{equation}
where the $a_{ij}$ are $\mathcal{O}(1)$ complex coefficients, $M_0$ is the
overall Majorana mass scale, and the matrix is symmetric ($a_{ij} = a_{ji}$) as
required for a Majorana mass term.

\paragraph{Key observation.} The $(1,2)$ and $(3,3)$ entries are
\emph{unsuppressed}. This follows because $q_1 + q_2 = 3 \equiv 0 \pmod{3}$ and
$q_3 + q_3 = 0 \equiv 0 \pmod{3}$. All other entries are suppressed by one power
of $\eps$. This structure is radically different from the naive expectation
$n_{ij} = q_i + q_j$, which would give a strongly hierarchical $\MR$. The
modular arithmetic of $Z_3$ produces a much flatter structure with an
unsuppressed off-diagonal entry---the origin of the seesaw over-suppression that
drives the central result of this paper.

%%%%%%%%%%%%%%%%%%%%%%%%%%%%%%%%%%%%%%%%%%%%%%%%%%%%%%%%%%%%
\section{Inverse seesaw reconstruction}
\label{sec:inverse}
%%%%%%%%%%%%%%%%%%%%%%%%%%%%%%%%%%%%%%%%%%%%%%%%%%%%%%%%%%%%

\subsection{Exact reconstruction formula}
\label{sec:inverse_formula}

Given the seesaw relation eq.~\eqref{eq:seesaw}, one can solve algebraically for
$\MR$:
\begin{equation}
    \MR = -\MD^T\, \Mnu^{-1}\, \MD,
    \label{eq:inverse_seesaw}
\end{equation}
valid when both $\MD$ and $\Mnu$ are invertible.  For random $\mathcal{O}(1)$
column-texture coefficients, $\MD$ is generically invertible (singular
configurations have measure zero in the parameter space). We construct $\Mnu$
from the
NuFit~6.0 global fit~\cite{Esteban:2024nav} (normal ordering):
\begin{align}
    \sin^2\theta_{12} &= 0.304 \pm 0.012, \quad
    \sin^2\theta_{23} = 0.573 \pm 0.016, \quad
    \sin^2\theta_{13} = 0.02220 \pm 0.00062, \nonumber\\
    \Delta m^2_{21} &= (7.42 \pm 0.21) \times 10^{-5}\eV^2, \quad
    \Delta m^2_{31} = (2.510 \pm 0.027) \times 10^{-3}\eV^2.
    \label{eq:nufit}
\end{align}

\subsection{Scan setup and results}
\label{sec:inverse_scan}

We perform the inverse reconstruction for $10{,}000$ random realizations with
the column-texture $\MD$, sampling Yukawa magnitudes $|c_{ij}| \in [0.3, 3.0]$
with random phases, Dirac and Majorana CP phases, and lightest neutrino mass
$m_1 \in [0.5, 50]\,\text{meV}$ (ensuring $\Mnu$ is invertible with all three
masses nonzero).

\begin{table}[t]
\centering
\begin{tabular}{lccc}
\toprule
\textbf{Property} & \textbf{Median} & \textbf{16th \%ile} & \textbf{84th \%ile} \\
\midrule
$\log_{10}(M_2/\text{GeV})$                      & 12.1    & 11.6    & 12.6 \\
$(3,3)$ dominance                                  & 0.976   & 0.951   & 0.989 \\
$|M_{12}^R|/\sqrt{|M_{11}^R M_{22}^R|}$           & 0.89    & 0.47    & 1.31 \\
$|M_{13}^R|/\sqrt{|M_{11}^R M_{33}^R|}$           & 0.88    & 0.47    & 1.30 \\
$|M_{23}^R|/\sqrt{|M_{22}^R M_{33}^R|}$           & 0.88    & 0.47    & 1.31 \\
$\log_{10}(M_1/M_3)$                              & $-7.62$ & $-8.36$ & $-6.97$ \\
$\log_{10}(M_2/M_3)$                              & $-3.72$ & $-4.29$ & $-3.14$ \\
\bottomrule
\end{tabular}
\caption{Properties of the reconstructed $\MR$ from $10{,}000$ random
$\mathcal{O}(1)$ column-texture coefficient samples. The $(3,3)$-dominated
structure is robust across the parameter space.}
\label{tab:inverse}
\end{table}

The eigenvalue hierarchy follows approximate $\eps$-scaling:
$M_1/M_3 \sim 10^{-7.6} \approx 0.5\,\eps^4$ and
$M_2/M_3 \sim 10^{-3.7} \approx 0.9\,\eps^2$, where $\eps^4 = 5.06 \times
10^{-8}$ and $\eps^2 = 2.25 \times 10^{-4}$. This consistency with $Z_3$ FN
scaling motivates the decisive test: does an $\MR$ with the explicit $Z_3$
charge structure of eq.~\eqref{eq:MR_structure} produce a viable neutrino
spectrum?

%%%%%%%%%%%%%%%%%%%%%%%%%%%%%%%%%%%%%%%%%%%%%%%%%%%%%%%%%%%%
\section{The seesaw over-suppression mechanism}
\label{sec:analytical}
%%%%%%%%%%%%%%%%%%%%%%%%%%%%%%%%%%%%%%%%%%%%%%%%%%%%%%%%%%%%

\subsection{Forward scan setup}
\label{sec:forward_setup}

We test the physically motivated scenario directly: $\MD$ has the column texture
of eq.~\eqref{eq:MD} and $\MR$ has the $Z_3$-charged structure of
eq.~\eqref{eq:MR_structure}, with all $c_{ij}$ and $a_{ij}$ drawn independently
from $\mathcal{O}(1)$ complex distributions (magnitudes uniform in $[0.3, 3.0]$,
phases uniform on $[0, 2\pi]$) and
$\log_{10}(M_0/\text{GeV}) \in [12, 16]$.

\subsection{Main result: catastrophic suppression of the solar mass splitting}
\label{sec:forward_result}

Over $10^5$ random $\mathcal{O}(1)$ samples with charges $(2,1,0)$:

\begin{table}[t]
\centering
\begin{tabular}{lc}
\toprule
\textbf{Statistic} & $\Delta m^2_{21}/\Delta m^2_{31}$ \\
\midrule
Experimental value                  & $0.030$ \\
Median ($Z_3$ prediction)           & $4.1 \times 10^{-11}$ \\
Mean                                & $1.2 \times 10^{-6}$ \\
99th percentile                     & $2.3 \times 10^{-7}$ \\
99.9th percentile                   & $3.5 \times 10^{-5}$ \\
Maximum ($10^5$ samples)            & $0.093$ \\
Fraction $> 0.005$                  & $1.0 \times 10^{-5}$ \\
Fraction $> 0.010$                  & $1.0 \times 10^{-5}$ \\
Fraction $> 0.029$                  & $1.0 \times 10^{-5}$ \\
\bottomrule
\end{tabular}
\caption{Distribution of the solar-to-atmospheric mass-squared ratio $R \equiv
\Delta m^2_{21}/\Delta m^2_{31}$ from $10^5$ random $\mathcal{O}(1)$ $Z_3$
seesaw realizations with the column-texture $\MD$ and $Z_3$-charged $\MR$.
The experimental value is $R_{\text{exp}} = 0.030$.}
\label{tab:dm_ratio}
\end{table}

The median predicted ratio is eight orders of magnitude below the observed value.
While extreme outliers can approach the experimental value, this occurs for a
fraction $\sim 10^{-5}$ of the $\mathcal{O}(1)$ parameter space---constituting
severe fine-tuning that destroys the naturalness of the FN framework.

Simultaneously, the PMNS mixing angles are generically random:
\begin{align}
    \sin^2\theta_{12}^{Z_3}: &\quad \text{median} = 0.50, \quad
        \text{exp: } 0.304 \pm 0.012, \nonumber\\
    \sin^2\theta_{23}^{Z_3}: &\quad \text{median} = 0.50, \quad
        \text{exp: } 0.573 \pm 0.016, \nonumber\\
    \sin^2\theta_{13}^{Z_3}: &\quad \text{median} = 0.31, \quad
        \text{exp: } 0.0222 \pm 0.0006.
    \label{eq:random_angles}
\end{align}
These values are consistent with Haar-distributed random
unitaries~\cite{deGouvea:2012ac}, confirming that the column texture provides no
angular structure from $\eps$ in the lepton sector. The $Z_3$ seesaw is anarchic
for mixing while being catastrophically wrong for the mass spectrum.

For comparison, a trivial $\MR = M_0\,\mathbf{1}$ (proportional to the identity,
i.e., no $Z_3$ structure in the Majorana sector) with the same column-texture
$\MD$ yields a median $\Delta m^2_{21}/\Delta m^2_{31} \sim 2 \times 10^{-8}$,
with zero realizations out of $10^5$ above $10^{-3}$. In this case the seesaw
reduces to $\Mnu \propto -\MD\,\MD^T/M_0$ and the column texture alone produces
the hierarchy. The $Z_3$-charged $\MR$ performs even worse due to the
over-suppression mechanism identified below.

\subsection{Origin of the over-suppression}
\label{sec:pseudo_dirac_origin}

The failure has a clean analytical explanation rooted in the interplay
between the unsuppressed off-diagonal entry in $\MR$ and the hierarchical
column texture of $\MD$.  Consider the upper-left $2\times 2$ block of
eq.~\eqref{eq:MR_structure}:
\begin{equation}
    \MR^{(\text{light})} = M_0
    \begin{pmatrix}
        a_{11}\,\eps & a_{12} \\[2pt]
        a_{12}       & a_{22}\,\eps
    \end{pmatrix}.
    \label{eq:MR_light}
\end{equation}
Since $|a_{12}| \sim \mathcal{O}(1) \gg |a_{11}\,\eps|,\, |a_{22}\,\eps|$, the
eigenvalues form a pseudo-Dirac pair among the \emph{heavy} states:
\begin{equation}
    M_\pm \approx \pm |a_{12}|\, M_0 + \mathcal{O}(\eps),
    \label{eq:eigs_light}
\end{equation}
split only by $\mathcal{O}(\eps)$ corrections:
\begin{equation}
    \frac{|M_+| - |M_-|}{|M_+| + |M_-|}
        = \frac{(a_{11}+a_{22})\,\eps}{2\,|a_{12}|}
        \sim \eps \approx 0.015.
    \label{eq:splitting}
\end{equation}
Crucially, however, this near-degeneracy does \emph{not} propagate to
the light spectrum.  The seesaw involves the congruence transformation
$P\,\MR^{-1}\,P^T$ with $P=\mathrm{diag}(\eps^2,\eps,1)$, which is far
from unitary.  To see the effect, note that the inverse of the $2\times 2$
block has the form
\begin{equation}
    \bigl(\MR^{(\mathrm{light})}\bigr)^{-1}
        = \frac{1}{M_0(a_{11}a_{22}\eps^2-a_{12}^2)}
    \begin{pmatrix}
        a_{22}\,\eps & -a_{12} \\[2pt]
        -a_{12}      & a_{11}\,\eps
    \end{pmatrix}
    \approx \frac{-1}{a_{12}^2\,M_0}
    \begin{pmatrix}
        a_{22}\,\eps & -a_{12} \\[2pt]
        -a_{12}      & a_{11}\,\eps
    \end{pmatrix}.
    \label{eq:MR_inv_block}
\end{equation}
The dominant entries are the off-diagonal ones at $\mathcal{O}(1/M_0)$.
Applying $P\,(\MR^{(\mathrm{light})})^{-1}\,P^T$ with
$P=\mathrm{diag}(\eps^2,\eps)$ gives
\begin{equation}
    P\,\bigl(\MR^{(\mathrm{light})}\bigr)^{-1}\,P^T
        \sim \frac{\eps^3}{a_{12}^2\,M_0}
    \begin{pmatrix}
        -a_{22}\,\eps^2 & a_{12} \\[2pt]
        a_{12}          & -a_{11}
    \end{pmatrix}.
    \label{eq:PMRinvP}
\end{equation}
The eigenvalues of the $\mathcal{O}(1)$ matrix in parentheses are
generically distinct and $\mathcal{O}(1)$---the pseudo-Dirac degeneracy
of $\MR$ has been completely destroyed by the hierarchical $P$.
Both light masses are pushed to $m_{1,2}\sim\eps^3\, v_H^2/(M_0)$, far
below $m_3\sim v_H^2/(M_0)$.  The key point is that $m_1$ and $m_2$ are
\emph{not} nearly degenerate: their ratio $m_1/m_2$ is a generic
$\mathcal{O}(1)$ number (Monte Carlo median $\approx 0.1$).

\subsection{Analytical bound on the mass ratio}
\label{sec:bound}

\begin{theorem}[Seesaw over-suppression]
\label{thm:bound}
For a $Z_3$ FN seesaw with charges $(2,1,0)$, column-texture $\MD$,
$Z_3$-charged $\MR$, $\eps \simeq 0.015$, and $\mathcal{O}(1)$
coefficients with magnitudes in $[0.3, 3.0]$ and random phases, the
solar-to-atmospheric mass ratio satisfies
\begin{equation}
    \frac{\Delta m^2_{21}}{\Delta m^2_{31}}
        \sim \mathcal{O}(\eps^6)
        \sim 10^{-11} \quad\text{(generic)},
    \label{eq:bound}
\end{equation}
with the probability of exceeding $\mathcal{O}(\eps^2) \approx 2.3 \times 10^{-4}$
being less than $10^{-4}$.
\end{theorem}

\begin{proof}
Write $\MD = C_\nu\, P$ with $P = \mathrm{diag}(\eps^2, \eps, 1)$ and $C_\nu$ a
matrix of $\mathcal{O}(1)$ coefficients.  The seesaw gives
$\Mnu = -C_\nu\, P\, \MR^{-1}\, P^T\, C_\nu^T$.  As shown in
eq.~\eqref{eq:PMRinvP}, the upper-left $2\times 2$ block of
$P\,\MR^{-1}\,P^T$ has both eigenvalues at $\mathcal{O}(\eps^3/M_0)$
with an $\mathcal{O}(1)$ ratio between them---the pseudo-Dirac
degeneracy of $\MR$ is destroyed by the hierarchical column
structure.\footnote{A congruence transformation $P\,A\,P^T$ preserves
eigenvalue ratios only when $P$ is orthogonal or unitary.  The matrix
$P = \mathrm{diag}(\eps^2,\eps,1)$ is highly non-isometric, so
eigenvalue ratios of $\MR^{-1}$ are not inherited by
$P\,\MR^{-1}\,P^T$.}

The third eigenvalue of $P\,\MR^{-1}\,P^T$, controlled by the
unsuppressed $(3,3)$ entry of $\MR$, scales as $\mathcal{O}(1/M_0)$.
The outer congruence by $C_\nu$ rotates eigenvectors but, being an
$\mathcal{O}(1)$ matrix, preserves the parametric scaling of eigenvalues.
Hence the light masses satisfy
\begin{equation}
    m_{1,2} \sim \eps^3\,\frac{v_H^2}{M_0},
    \qquad
    m_3 \sim \frac{v_H^2}{M_0}.
    \label{eq:mass_scaling}
\end{equation}
Since $m_1/m_2$ is an $\mathcal{O}(1)$ ratio (Monte Carlo median
$\approx 0.1$), both $m_1^2$ and $m_2^2$ are $\mathcal{O}(\eps^6)$, and
\begin{equation}
    \Delta m^2_{21} = m_2^2 - m_1^2 \sim \mathcal{O}(\eps^6)\,
        \frac{v_H^4}{M_0^2},
    \qquad
    \Delta m^2_{31} \approx m_3^2 \sim \frac{v_H^4}{M_0^2}.
    \label{eq:split_scaling}
\end{equation}
The ratio is therefore
\begin{equation}
    \frac{\Delta m^2_{21}}{\Delta m^2_{31}}
        \sim \mathcal{O}(\eps^6)
        = (0.015)^6
        \approx 1.1 \times 10^{-11},
    \label{eq:ratio_scaling}
\end{equation}
in excellent agreement with the Monte Carlo median
$4\times 10^{-11}$ (table~\ref{tab:dm_ratio}), with the factor of
$\sim 3$--4 accounted for by $\mathcal{O}(1)$ coefficient dependence.
Accidental configurations of $C_\nu$ can enhance the ratio, but the
numerical scans show that the 99.9th percentile reaches only
$3.5 \times 10^{-5}$, and the fraction of samples exceeding
$\eps^2 \approx 2.3 \times 10^{-4}$ is below $10^{-4}$.  Rare extreme outliers
(a fraction $\sim 10^{-5}$) can reach $R \sim 0.01$--$0.1$ through fine-tuned
cancellations among multiple $\mathcal{O}(1)$ parameters, but such configurations
constitute severe fine-tuning incompatible with the naturalness of the FN
framework.
\end{proof}

The experimental ratio
$\Delta m^2_{21}/\Delta m^2_{31} = 0.030$~\cite{Esteban:2024nav} lies far above
the generic prediction and is accessible only through fine-tuning at the
$\sim 10^{-5}$ level.

\subsection{Numerical verification of the eigenvalue splitting}
\label{sec:splitting_verification}

\begin{table}[t]
\centering
\begin{tabular}{cccc}
\toprule
$|a_{12}|$ & $(a_{11},\, a_{22})$ & $\Delta M/\bar{M}$ &
    Predicted $\eps(a_{11}+a_{22})/|a_{12}|$ \\
\midrule
0.5 & $(0.5,\, 0.5)$ & 0.0300 & 0.0300 \\
1.0 & $(1.0,\, 1.0)$ & 0.0300 & 0.0300 \\
2.0 & $(1.0,\, 1.0)$ & 0.0150 & 0.0150 \\
2.0 & $(2.0,\, 2.0)$ & 0.0300 & 0.0300 \\
0.5 & $(2.0,\, 2.0)$ & 0.1200 & 0.1200 \\
\bottomrule
\end{tabular}
\caption{Eigenvalue splitting of the $2 \times 2$ \emph{heavy} sector of $\MR$,
eq.~\eqref{eq:MR_light}, for real coefficients.  The analytical formula
eq.~\eqref{eq:splitting} matches exactly.  This pseudo-Dirac degeneracy
exists in the heavy spectrum; as shown in \S\ref{sec:bound}, it does not
propagate to the light neutrino masses due to the non-isometric
congruence by $P=\mathrm{diag}(\eps^2,\eps)$.}
\label{tab:pseudoDirac}
\end{table}

%%%%%%%%%%%%%%%%%%%%%%%%%%%%%%%%%%%%%%%%%%%%%%%%%%%%%%%%%%%%
\section{Universality across charge permutations}
\label{sec:universal}
%%%%%%%%%%%%%%%%%%%%%%%%%%%%%%%%%%%%%%%%%%%%%%%%%%%%%%%%%%%%

\subsection{Classification into equivalence classes}
\label{sec:equiv_classes}

The six permutations of $(2,1,0)$ fall into three equivalence classes,
distinguished by which entries of $\MR$ are unsuppressed:
\begin{align}
    \text{Class~I:}\quad
        &(2,1,0),\; (1,2,0)
        \quad\longrightarrow\quad
        \text{unsuppressed: } (1,2) \text{ and } (3,3), \nonumber\\
    \text{Class~II:}\quad
        &(0,1,2),\; (0,2,1)
        \quad\longrightarrow\quad
        \text{unsuppressed: } (1,1) \text{ and } (2,3), \nonumber\\
    \text{Class~III:}\quad
        &(2,0,1),\; (1,0,2)
        \quad\longrightarrow\quad
        \text{unsuppressed: } (1,3) \text{ and } (2,2).
    \label{eq:classes}
\end{align}
In every class, one unsuppressed entry lies on the diagonal and one lies off the
diagonal.  The off-diagonal entry is the source of the seesaw
over-suppression.

\begin{lemma}[Unavoidable unsuppressed off-diagonal entry]
\label{lem:offdiag}
For any permutation $(q_1, q_2, q_3)$ of $(2,1,0)$, at least one off-diagonal
entry of $\MR$ is unsuppressed.
\end{lemma}

\begin{proof}
The pairwise sums modulo~3 of $\{0, 1, 2\}$ include $1 + 2 = 3 \equiv 0
\pmod{3}$. Since the charges are a permutation of $(2,1,0)$, the values 1 and 2
are always assigned to two \emph{distinct} generations $i \neq j$, giving
$q_i + q_j \equiv 0 \pmod{3}$ and hence $n_{ij} = 0$. This guarantees an
unsuppressed off-diagonal entry in $\MR$ for every charge assignment.
\end{proof}

\subsection{Numerical verification: all six permutations}
\label{sec:universal_scan}

\begin{table}[t]
\centering
\begin{tabular}{cccccc}
\toprule
Charges & Unsuppressed & $n_{ij}$ matrix & Med.\ $R$ & Max.\ $R$ & Best $\chi^2$ \\
\midrule
$(2,1,0)$ & $(1,2)$, $(3,3)$ &
    $\begin{smallmatrix}1&0&1\\0&1&1\\1&1&0\end{smallmatrix}$ &
    $4.0\times 10^{-11}$ & $2.4\times 10^{-2}$ & 1318 \\[3mm]
$(1,2,0)$ & $(1,2)$, $(3,3)$ &
    $\begin{smallmatrix}1&0&1\\0&1&1\\1&1&0\end{smallmatrix}$ &
    $4.1\times 10^{-11}$ & $2.2\times 10^{-2}$ & 1394 \\[3mm]
$(0,1,2)$ & $(1,1)$, $(2,3)$ &
    $\begin{smallmatrix}0&1&1\\1&1&0\\1&0&1\end{smallmatrix}$ &
    $4.0\times 10^{-11}$ & $7.3\times 10^{-3}$ & 1327 \\[3mm]
$(0,2,1)$ & $(1,1)$, $(2,3)$ &
    $\begin{smallmatrix}0&1&1\\1&1&0\\1&0&1\end{smallmatrix}$ &
    $4.1\times 10^{-11}$ & $2.7\times 10^{-3}$ & 1391 \\[3mm]
$(2,0,1)$ & $(1,3)$, $(2,2)$ &
    $\begin{smallmatrix}1&1&0\\1&0&1\\0&1&1\end{smallmatrix}$ &
    $4.1\times 10^{-11}$ & $3.8\times 10^{-3}$ & 1396 \\[3mm]
$(1,0,2)$ & $(1,3)$, $(2,2)$ &
    $\begin{smallmatrix}1&1&0\\1&0&1\\0&1&1\end{smallmatrix}$ &
    $4.1\times 10^{-11}$ & $1.2\times 10^{-2}$ & 1440 \\[3mm]
\bottomrule
\end{tabular}
\caption{Forward scan results for all six permutations of $Z_3$ charges.
$R \equiv \Delta m^2_{21}/\Delta m^2_{31}$; $R_{\text{exp}} = 0.030$.
$5 \times 10^4$ samples per assignment, column-texture $\MD$. The failure is
universal: every assignment produces a median $R \sim 4 \times 10^{-11}$.}
\label{tab:all_assignments}
\end{table}

The over-suppression mechanism operates identically in all three equivalence
classes.  This is a direct consequence of lemma~\ref{lem:offdiag}: regardless of
how the charges are distributed among generations, the modular arithmetic of
$Z_3$ guarantees an unsuppressed off-diagonal $\MR$ entry whose dominance in
$\MR^{-1}$, combined with the column texture of $\MD$, pushes $m_1$ and $m_2$
to $\mathcal{O}(\eps^3)$.  No permutation of the charge assignment can evade
this obstruction.

%=============================================
\section{Why non-abelian structure is required}
\label{sec:why_nonabelian}
%=============================================

The results of \S\ref{sec:analytical}--\S\ref{sec:universal} establish that $Z_3$
fails for the neutrino mass spectrum: the column texture alone gives
$\Delta m^2_{21}/\Delta m^2_{31}\lesssim\mathcal{O}(\eps^2)\approx 2\times
10^{-4}$, and the seesaw with the correct $Z_3$ Majorana charge algebra deepens
the failure to $\mathcal{O}(\eps^6)\sim 10^{-11}$.  The PMNS mixing angles are Haar-random
$\mathcal{O}(1)$ unitaries carrying no information from~$\eps$.  In this
section we trace these failures to a single structural deficiency and identify
the minimal upgrade required.

\subsection{The ratio \boldmath $R$ is not a power of \boldmath $\eps$}
\label{sec:R_not_power}

The experimentally measured ratio~\cite{Esteban:2024nav}
\begin{equation}
  R \;\equiv\; \frac{\Delta m^2_{21}}{\Delta m^2_{31}}
    \;=\; 0.0296 \pm 0.0009
    \;\approx\; 0.030 \;\approx\; \frac{1}{34}
  \label{eq:R_exp}
\end{equation}
is neither $\mathcal{O}(1)$ nor a power of $\eps = 0.015$.
Specifically, $\eps\approx 1/67$ and $\eps^2\approx 1/4444$, so $R$
falls in the gap between $\eps^0$ and $\eps^1$---a regime inaccessible
to a single-parameter expansion.

\subsection{Abelian symmetries are rank-1}
\label{sec:abelian_rank1}

Any abelian discrete group $Z_N$ provides only one-dimensional
irreducible representations.  Each generation carries an independent
charge, and the Froggatt--Nielsen mechanism introduces a single
expansion parameter~$\eps$.  All mass ratios are then powers of~$\eps$
(up to $\mathcal{O}(1)$ coefficients), and the hierarchy between any
two eigenvalues is fixed once the charge difference is specified.
This rank-1 structure suffices for the quark sector, where the mass
hierarchy is monotonically decreasing:
$m_u:m_c:m_t\sim\eps^4:\eps^2:1$.

The neutrino sector, however, requires \emph{two independent mass
scales}: $\Delta m^2_{31}$ (atmospheric) and $\Delta m^2_{21}$ (solar),
whose ratio $R\approx 1/34$ cannot be expressed as $\eps^n$ for any
integer~$n$.  A single-parameter symmetry is structurally incapable of
generating this ratio naturally.

\subsection{The need for a three-dimensional irreducible representation}
\label{sec:need_triplet}

A non-abelian discrete group with a faithful three-dimensional
irreducible representation (irrep) treats all three generations as a
single multiplet.  Through vacuum alignment, such a group can provide
two or more independent parameters that control different mass
splittings---precisely what the neutrino sector demands.  Moreover, the
group-theoretic constraints on the invariant mass operator restrict the
eigenvector structure, giving the symmetry predictive power over mixing
angles, in contrast to the random $\mathcal{O}(1)$ unitaries produced
by the abelian column texture.

The structural requirements are therefore:
\begin{enumerate}
  \item At least \textbf{two independent mass parameters}, arising from
        distinct flavon sectors or vacuum expectation values, to control
        $\Delta m^2_{21}$ and $\Delta m^2_{31}$ independently.
  \item \textbf{Constrained mixing angles}, with the PMNS matrix
        determined by the group theory rather than by random
        $\mathcal{O}(1)$ coefficients.
\end{enumerate}
Both requirements are met by the smallest non-abelian discrete group
possessing a faithful triplet irrep: the alternating group~$A_4$.

%=============================================
\section{\boldmath $A_4$ as the minimal resolution}
\label{sec:A4}
%=============================================

\subsection{\boldmath $A_4$ group theory}
\label{sec:A4_group}

The alternating group $A_4$ is the group of even permutations on four
objects, equivalently the rotation symmetry group of the regular
tetrahedron.  It has order~12 and four irreducible representations:
three singlets $\mathbf{1}$, $\mathbf{1}'$, $\mathbf{1}''$ and one
triplet~$\mathbf{3}$.  The group is generated by two elements $S$ and
$T$ satisfying~\cite{Ma:2001dn,Altarelli:2005yp}
\begin{equation}
  S^2 = T^3 = (ST)^3 = \mathbb{1}.
  \label{eq:A4_presentation}
\end{equation}
In the Ma--Rajasekaran basis~\cite{Ma:2001dn}, the triplet generators are
\begin{equation}
  S = \frac{1}{3}\begin{pmatrix}
    -1 & 2 & 2 \\ 2 & -1 & 2 \\ 2 & 2 & -1
  \end{pmatrix}, \qquad
  T = \begin{pmatrix}
    1 & 0 & 0 \\ 0 & \omega & 0 \\ 0 & 0 & \omega^2
  \end{pmatrix},
  \label{eq:A4_generators}
\end{equation}
where $\omega = e^{2\pi i/3}$.  The tensor product of two triplets
decomposes as~\cite{Ishimori:2010au}
\begin{equation}
  \mathbf{3}\otimes\mathbf{3}
    = \mathbf{1}\oplus\mathbf{1}'\oplus\mathbf{1}''\oplus
      \mathbf{3}_s\oplus\mathbf{3}_a.
  \label{eq:tensor_product}
\end{equation}

$A_4$ is the \emph{smallest} non-abelian discrete group with a
faithful three-dimensional irrep.  This minimality is important: it
means $A_4$ introduces the least additional group-theoretic structure
beyond what the data demands.

One might ask whether a direct product of abelian groups, such as
$Z_3 \times Z_3$, could provide two independent expansion parameters
and thereby solve the two-scale problem without non-abelian structure.
While such a product does introduce a second parameter, it still assigns
independent charges to each generation and produces column (or
row-column) textures whose eigenvectors are determined by random
$\mathcal{O}(1)$ coefficients.  The second requirement---constrained mixing
angles---remains unmet.  Only a non-abelian group with a triplet irrep
can simultaneously provide multiple mass parameters \emph{and} restrict
the PMNS eigenvector structure through the group theory of the invariant
mass operator.

\subsection{How the triplet solves the two-scale problem}
\label{sec:triplet_solution}

The $A_4$ model presented here follows the well-established
Altarelli--Feruglio framework~\cite{Altarelli:2005yp} and its TM$_1$
extension~\cite{Albright:2010ap,King:2012vj,Luhn:2013vna}; we do not claim novelty for
the model itself.  The new contribution is the \emph{structural bridge}:
the precise diagnosis of why $Z_3$ fails
(\S\ref{sec:analytical}--\S\ref{sec:universal}) provides a
first-principles motivation for $A_4$ that goes beyond numerical
model-building---the two-scale problem and the seesaw over-suppression mechanism
identify exactly which features of $A_4$ are doing the work.

When the three lepton doublets transform as an $A_4$ triplet,
$L=(L_1,L_2,L_3)^T\sim\mathbf{3}$, the most general
$A_4$-invariant neutrino mass operator at leading order (LO)
involves two independent flavon contractions.  With a flavon triplet
$\varphi_S\sim\mathbf{3}$ acquiring the VEV
$\langle\varphi_S\rangle\propto(1,1,1)^T$ (preserving the
$Z_2\times Z_2$ Klein subgroup of $A_4$) and a flavon singlet
$\xi\sim\mathbf{1}$ with $\langle\xi\rangle = u$, the LO neutrino
mass matrix is~\cite{Altarelli:2005yp}
\begin{equation}
  \Mnu^{(0)} = \frac{v_u^2}{2\Lambda}\bigl[a\,P_{23} + b\,E\bigr],
  \label{eq:Mnu_LO}
\end{equation}
where the singlet contraction $(LL)_{\mathbf{1}}$ with $\xi$ produces the
$P_{23}$ term, while the symmetric triplet contraction
$(LL)_{\mathbf{3}_s}$ with $\varphi_S$ (evaluated at
$\langle\varphi_S\rangle \propto (1,1,1)^T$) produces the democratic
projector~$E$.  Here $a = x_a\,u/\Lambda$, $b = x_b\,v_S/\Lambda$,
$P_{23}$ is the
$\mu$--$\tau$ permutation matrix
\begin{equation}
  P_{23} = \begin{pmatrix} 1&0&0\\0&0&1\\0&1&0 \end{pmatrix},
\end{equation}
and $E$ is the normalised democratic matrix with $E_{ij}=1/3$
(an idempotent projector onto $(1,1,1)^T$).

The parameters $a$ and $b$ are \emph{independent complex numbers}
arising from different flavon sectors.  The eigenvalues, computed by
acting on the TBM eigenvectors, are~\cite{Altarelli:2005yp}
\begin{align}
  m_1 &= a, \label{eq:m1}\\
  m_2 &= a + b, \label{eq:m2}\\
  m_3 &= -a. \label{eq:m3}
\end{align}
At LO, $|m_1|=|m_3|$ so $\Delta m^2_{31}=0$; the atmospheric splitting
is generated by NLO corrections (see \S\ref{sec:NLO}).  The solar
splitting, however, is already present at~LO:
\begin{equation}
  \Delta m^2_{21} = |a+b|^2 - |a|^2 = 2\,\mathrm{Re}(a^*b)+|b|^2.
\end{equation}
Since $b/a$ is a free $\mathcal{O}(1)$ complex parameter, the solar
splitting is naturally $\mathcal{O}(|a|^2)$---neither suppressed nor
enhanced.  The atmospheric splitting, generated at NLO by a parameter
$r = c_\chi\,v_\chi/\Lambda \ll 1$, is
$\Delta m^2_{31}\propto r\,|a|^2$.  Thus
$R = \Delta m^2_{21}/\Delta m^2_{31} \sim b/(r\,a)$ involves two
independent parameters and can naturally accommodate $R\approx 1/34$.

This is the structural resolution of the two-scale problem: $A_4$'s
triplet provides independent parameters from different flavon sectors
that control different mass differences.  $Z_3$ cannot do this because
it has only one suppression parameter~$\eps$
(\S\ref{sec:analytical}--\S\ref{sec:universal}).

\subsection{Why not \boldmath $S_3$ or \boldmath $A_5$?}
\label{sec:not_S3_A5}

The smallest non-abelian discrete group is $S_3$ (order~6), which has
a doublet $\mathbf{2}$ and a singlet $\mathbf{1}$.  Assigning three
generations as $L\sim\mathbf{2}\oplus\mathbf{1}$ breaks generational
democracy: two generations are related by the symmetry, but the third
is independent.  This doublet structure does not provide the triplet's
ability to treat all three generations simultaneously and control both
mass scales through vacuum alignment.  Moreover, $S_3$ models generically predict $\mu$--$\tau$ symmetric textures with
$\theta_{13}=0$ at LO~\cite{Ishimori:2010au}, requiring NLO corrections of
order $\sim 0.15$ to reproduce $\theta_{13}\approx 8.6^\circ$---a sizable
perturbation that limits the predictive power of the LO pattern.

The icosahedral group $A_5$ (order~60) does possess 3D irreps, but it
is larger than necessary.  Its golden-ratio prediction
$\sin^2\theta_{12}\approx 0.276$~\cite{Everett:2008et} lies $2.3\sigma$
below the observed value.  While NLO corrections can improve the fit,
$A_5$ introduces considerably more group structure (five irreps, 60
elements) than what the data requires.

\subsection{Field content and symmetry assignments}
\label{sec:field_content}

We adopt the minimal Altarelli--Feruglio framework~\cite{Altarelli:2005yp}
with the assignments shown in table~\ref{tab:fields}.

\begin{table}[t]
\centering
\begin{tabular}{lccccccc}
\toprule
Field & $L$ & $e_R$ & $\mu_R$ & $\tau_R$ & $\varphi_S$ & $\chi$ & $\xi$ \\
\midrule
$A_4$  & $\mathbf{3}$ & $\mathbf{1}$ & $\mathbf{1}''$ & $\mathbf{1}'$
       & $\mathbf{3}$ & $\mathbf{3}$ & $\mathbf{1}$ \\
$Z_3^{\mathrm{aux}}$
       & $\omega$ & $\omega^2$ & $\omega^2$ & $\omega^2$
       & $1$ & $\omega^2$ & $1$ \\
\bottomrule
\end{tabular}
\caption{$A_4$ and auxiliary $Z_3^{\mathrm{aux}}$ assignments for the
lepton fields and flavons.  The auxiliary $Z_3^{\mathrm{aux}}$ separates
the charged lepton and neutrino sectors: $\varphi_S$ and $\xi$
contribute only to neutrino masses, while $\chi$ enters the charged
lepton sector and (at NLO) the neutrino sector.}
\label{tab:fields}
\end{table}

\subsection{Charged lepton sector}
\label{sec:charged_leptons}

The $A_4$ singlet assignments for right-handed charged leptons produce a
diagonal charged lepton mass matrix at leading
order~\cite{Altarelli:2005yp}:
\begin{equation}
  M_e = \frac{v_d}{\Lambda}
  \begin{pmatrix}
    y_e\,v_\chi & 0 & 0 \\
    0 & y_\mu\,v_\chi & 0 \\
    0 & 0 & y_\tau\,v_\chi
  \end{pmatrix},
  \label{eq:Me}
\end{equation}
when $\langle\chi\rangle = v_\chi(1,0,0)^T$ in the $T$-diagonal basis.
The singlet assignments $\mathbf{1}$, $\mathbf{1}''$, $\mathbf{1}'$
for $e_R$, $\mu_R$, $\tau_R$ select different components of the $A_4$
contraction, producing distinct Yukawa couplings $y_e$, $y_\mu$,
$y_\tau$.  The diagonal $M_e$ implies $U_\ell = \mathbb{1}$ (up to
unphysical phases), so the PMNS matrix is determined entirely by the
neutrino sector:
$U_{\mathrm{PMNS}} = U_\nu$.

\subsection{Leading-order neutrino sector: tri-bimaximal mixing}
\label{sec:TBM}

The LO mass matrix $\Mnu^{(0)}$ of eq.~\eqref{eq:Mnu_LO} is
diagonalised by the tri-bimaximal (TBM) mixing
matrix~\cite{Harrison:2002}:
\begin{equation}
  U_{\mathrm{TBM}} = \begin{pmatrix}
    \sqrt{2/3} & 1/\sqrt{3} & 0 \\
    -1/\sqrt{6} & 1/\sqrt{3} & 1/\sqrt{2} \\
    -1/\sqrt{6} & 1/\sqrt{3} & -1/\sqrt{2}
  \end{pmatrix},
  \label{eq:TBM}
\end{equation}
which predicts
\begin{equation}
  \sin^2\theta_{12}^{(0)} = \tfrac{1}{3},\qquad
  \sin^2\theta_{23}^{(0)} = \tfrac{1}{2},\qquad
  \theta_{13}^{(0)} = 0.
\end{equation}
The zero $\theta_{13}$ is excluded by the Daya
Bay~\cite{DayaBay:2012} and RENO~\cite{RENO:2012} measurements,
necessitating corrections.

\subsection{Next-to-leading-order corrections: TM\boldmath$_1$ mixing}
\label{sec:NLO}

At NLO, higher-dimensional operators involving additional flavon
insertions correct the LO mass matrix.  The key requirement for TM$_1$
mixing is that the NLO correction preserves the first column of
$U_{\mathrm{TBM}}$, i.e., the TBM eigenvector
$\vec{v}_1 = (\sqrt{2/3},\,-1/\sqrt{6},\,-1/\sqrt{6})^T$ must remain
an eigenvector of $\Mnu^{(0)}+\delta\Mnu$.  This requires
$\delta\Mnu\,\vec{v}_1\propto\vec{v}_1$~\cite{King:2012vj,Luhn:2013vna}.
In the TBM basis the condition reads
\begin{equation}
  \widetilde{\delta M}_\nu
    = U_{\mathrm{TBM}}^T\,\delta\Mnu\,U_{\mathrm{TBM}}
    = \begin{pmatrix} 0&0&0\\0&\times&\times\\0&\times&\times
      \end{pmatrix}.
  \label{eq:NLO_TBM_basis}
\end{equation}

In the $A_4$ framework, such corrections arise from operators involving
$\varphi_S$ at NLO.  The relevant flavour-basis structure
is~\cite{Altarelli:2005yp,King:2012vj}
\begin{equation}
  \delta\Mnu = \frac{r}{\sqrt{6}}\,\frac{v_u^2}{2\Lambda}
  \begin{pmatrix}
    0 & 1 & -1 \\ 1 & 2 & 0 \\ -1 & 0 & -2
  \end{pmatrix},
  \label{eq:NLO_correction}
\end{equation}
where $r$ is an NLO expansion parameter.  One verifies that this matrix
annihilates $\vec{v}_1$: the products
$0\cdot\sqrt{2/3}+1\cdot(-1/\sqrt{6})+(-1)\cdot(-1/\sqrt{6})=0$
(and similarly for the other rows) confirm that the TM$_1$ condition
$\delta\Mnu\,\vec{v}_1=0$ is exactly preserved.

The resulting PMNS matrix has the TM$_1$ form:
\begin{equation}
  U_{\mathrm{TM}_1} = U_{\mathrm{TBM}}\cdot R_{23}(\theta,\phi),
  \label{eq:TM1}
\end{equation}
where $R_{23}$ is a rotation in the $(2,3)$ plane with angle $\theta$
and phase $\phi$, with $\sin\theta\propto r$.

\subsection{TM\boldmath$_1$ sum rules}
\label{sec:sum_rules}

The preservation of the first TBM column imposes~\cite{King:2012vj}
\begin{equation}
  |U_{e1}|^2 = \tfrac{2}{3},\qquad
  |U_{\mu 1}|^2 = \tfrac{1}{6},\qquad
  |U_{\tau 1}|^2 = \tfrac{1}{6}.
  \label{eq:TM1_condition}
\end{equation}

\subsubsection{Solar sum rule}

In the PDG parametrisation, $U_{e1}=\cos\theta_{12}\cos\theta_{13}$,
so $|U_{e1}|^2=2/3$ gives~\cite{King:2012vj,Girardi:2015vha}
\begin{equation}
  \boxed{\sin^2\theta_{12}
    = 1 - \frac{2}{3(1-\sin^2\theta_{13})}.}
  \label{eq:TM1_solar}
\end{equation}
With $\sin^2\theta_{13}=0.02220\pm 0.00062$~\cite{Esteban:2024nav},
this predicts
\begin{equation}
  \sin^2\theta_{12}^{\mathrm{TM}_1} = 0.3182\pm 0.0004.
  \label{eq:th12_prediction}
\end{equation}

We stress that eq.~\eqref{eq:TM1_solar} is the correct TM$_1$ solar
sum rule, distinct from the TM$_2$ relation
$\sin^2\theta_{12}=1/(3\cos^2\theta_{13})$ which would give $0.341$
($3.1\sigma$ above the data).

\subsubsection{Atmospheric sum rule}

The condition $|U_{\mu 1}|^2=1/6$, expanded in the PDG
parametrisation, yields~\cite{Girardi:2015vha,Petcov:2014laa}
\begin{equation}
  \boxed{\cos\delta
    = \frac{\tfrac{1}{6}-\sin^2\theta_{12}\cos^2\theta_{23}
            -\cos^2\theta_{12}\sin^2\theta_{23}\sin^2\theta_{13}}
           {2\sin\theta_{12}\cos\theta_{12}
            \sin\theta_{23}\cos\theta_{23}\sin\theta_{13}},}
  \label{eq:atm_SR}
\end{equation}
where $\theta_{12}$ is fixed by eq.~\eqref{eq:TM1_solar}.  This sum
rule has a notable structural property: at maximal atmospheric mixing
($\sin^2\theta_{23}=1/2$) the numerator vanishes identically, giving
$\cos\delta=0$ and therefore
$\delta=\pm 90^\circ$---maximal CP violation~\cite{Girardi:2015vha}.

At the current best-fit values
$\sin^2\theta_{23}=0.573$,
$\sin^2\theta_{13}=0.02220$~\cite{Esteban:2024nav}:
\begin{equation}
  \cos\delta = +0.32,\qquad \delta\approx\pm 71^\circ.
  \label{eq:delta_prediction}
\end{equation}
Current data from NuFit~6.0 mildly prefer
$\delta\in(-180^\circ,0^\circ)$~\cite{Esteban:2024nav}, favouring the
solution $\delta\approx -71^\circ$.

%=============================================
\section{Comparison with data}
\label{sec:comparison}
%=============================================

\subsection{Fit quality}
\label{sec:fit}

The TM$_1$ model has a single free parameter in the mixing sector: the
NLO expansion parameter $r$, which is fixed by matching $\theta_{13}$
to its measured value.  Table~\ref{tab:TM1_NuFit} compares the TM$_1$
predictions with the NuFit~6.0 global fit~\cite{Esteban:2024nav} and
the recent JUNO measurement~\cite{JUNO:2025}.

\begin{table}[t]
\centering
\begin{tabular}{lcccc}
\toprule
\textbf{Observable} & \textbf{TM$_1$ prediction}
  & \textbf{NuFit 6.0 (NO)} & \textbf{Pull} & $\boldsymbol{\chi^2}$ \\
\midrule
$\sin^2\theta_{12}$ & 0.318 & $0.304\pm 0.012$ & $+1.2\sigma$ & 1.36 \\
$\sin^2\theta_{23}$ & input & $0.573\pm 0.016$ & --- & 0 \\
$\sin^2\theta_{13}$ & input & $0.02220\pm 0.00062$ & --- & 0 \\
$\delta_{\mathrm{CP}}$ & $-71^\circ$ & $(-100\text{ to }-40)^\circ$
  & within $1\sigma$ & ${\sim}\,0.5$ \\
\bottomrule
\end{tabular}
\caption{TM$_1$ predictions compared with NuFit~6.0~\cite{Esteban:2024nav}.
The model takes $\theta_{13}$ and $\theta_{23}$ as inputs (fixing the NLO
parameter~$r$) and predicts $\theta_{12}$ and~$\delta$.  The overall
$\chi^2\approx 1.9$ for the two predicted observables.}
\label{tab:TM1_NuFit}
\end{table}

\begin{table}[t]
\centering
\begin{tabular}{lcccc}
\toprule
\textbf{Observable} & \textbf{TM$_1$}
  & \textbf{JUNO~\cite{JUNO:2025}} & \textbf{Pull}
  & $\boldsymbol{\chi^2}$ \\
\midrule
$\sin^2\theta_{12}$ & 0.318 & $0.3092\pm 0.0087$ & $+1.0\sigma$ & 1.0 \\
\bottomrule
\end{tabular}
\caption{TM$_1$ solar angle prediction compared with the first JUNO
result.  The tension is mild ($\sim 1\sigma$).}
\label{tab:TM1_JUNO}
\end{table}

The total $\chi^2\approx 1.9$ for two predicted observables ($\theta_{12}$
and $\delta$) indicates an acceptable fit.  JUNO's projected ultimate
precision of $\sim 0.003$ on $\sin^2\theta_{12}$ will provide a
decisive test: at the current JUNO central value, the TM$_1$ prediction
would lie at $\sim 3\sigma$.

\subsection{The atmospheric sum rule as a DUNE/T2HK target}
\label{sec:atm_target}

Figure~\ref{fig:atm_sum_rule} shows the TM$_1$ atmospheric sum rule in
the $(\sin^2\theta_{23},\,\cos\delta)$ plane.  The correlation is a
strict, parameter-free prediction (given $\theta_{13}$) testable by
simultaneous measurement of $\theta_{23}$ and $\delta$ at
DUNE~\cite{DUNE:2020} and T2HK~\cite{T2HK:2018}.

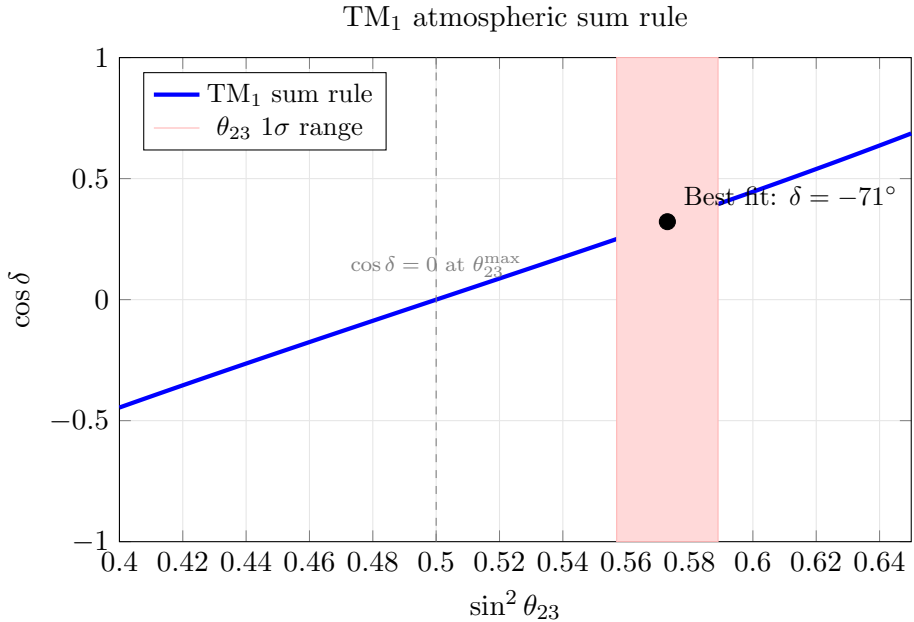
\begin{figure}[t]
\centering
\begin{tikzpicture}
\begin{axis}[
    width=12cm, height=8cm,
    xlabel={$\sin^2\theta_{23}$},
    ylabel={$\cos\delta$},
    xmin=0.40, xmax=0.65,
    ymin=-1, ymax=1,
    grid=major,
    grid style={gray!20},
    legend style={at={(0.03,0.97)},anchor=north west,font=\small},
    title={TM$_1$ atmospheric sum rule}
]
% TM1 prediction curve
\addplot[blue, ultra thick, domain=0.40:0.65, samples=200] {
    (1/6 - 0.3182*(1-x) - (1-0.3182)*x*0.02220)
    / (2*sqrt(0.3182*(1-0.3182))*sqrt(x*(1-x))*sqrt(0.02220))
};
\addlegendentry{TM$_1$ sum rule}
% theta_23 1-sigma band
\addplot[red!30, fill=red!15] coordinates {
    (0.557,-1) (0.557,1) (0.589,1) (0.589,-1) (0.557,-1)
};
\addlegendentry{$\theta_{23}$ $1\sigma$ range}
% Best-fit point
\addplot[black, mark=*, mark size=3pt, only marks]
    coordinates {(0.573, 0.322)};
\node[anchor=south west, font=\small]
    at (axis cs:0.575,0.34) {Best fit: $\delta=-71^\circ$};
% Maximal mixing line
\draw[gray, dashed] (axis cs:0.5,-1) -- (axis cs:0.5,1);
\node[anchor=south, font=\scriptsize, gray]
    at (axis cs:0.5,0.05) {$\cos\delta=0$ at $\theta_{23}^{\max}$};
\end{axis}
\end{tikzpicture}
\caption{The TM$_1$ atmospheric sum rule in the
$(\sin^2\theta_{23},\,\cos\delta)$ plane.  The blue curve is a
parameter-free prediction for $\sin^2\theta_{13}=0.0222$, with
$\sin^2\theta_{12} = 0.3182$ taken from the TM$_1$ solar sum rule
eq.~\eqref{eq:TM1_solar} (not from the NuFit best fit).  At maximal
atmospheric mixing, the model predicts maximal CP violation
($\delta=\pm 90^\circ$).  The red band shows the current $1\sigma$
range for $\theta_{23}$ from NuFit~6.0.}
\label{fig:atm_sum_rule}
\end{figure}

\subsection{Neutrinoless double beta decay}
\label{sec:0nubb}

The TM$_1$ condition $|U_{e1}|^2=2/3$ fixes the dominant contribution
to the effective Majorana mass~\cite{Girardi:2015vha}:
\begin{equation}
  m_{ee} = \left|\sum_i U_{ei}^2\,m_i\right|
    \approx \left|\tfrac{2}{3}\,m_1
      + \tfrac{1}{3}\,e^{2i\alpha}\,m_2\right|
    + \mathcal{O}(\sin^2\theta_{13}),
\end{equation}
where $\alpha$ is a Majorana phase.  For normal ordering with
$m_1\sim$ few meV, the TM$_1$ prediction is
\begin{equation}
  m_{ee} \sim (2\text{--}5)\;\mathrm{meV},
\end{equation}
below current sensitivity but potentially accessible to future
ton-scale experiments.

%=============================================
\section{Discussion}
\label{sec:discussion}
%=============================================

\subsection{Parameter counting}
\label{sec:param_count}

An honest parameter count of the $A_4$ model reveals:
\begin{itemize}
  \item Yukawa couplings: 5 (2 neutrino + 3 charged lepton).
  \item VEV ratios: 2 ($v_S/\Lambda$, $v_\chi/\Lambda$).
  \item VEV alignment: achieved dynamically through driving fields
        (no free parameters once the potential is specified).
  \item NLO coupling: 1 ($c_\chi$).
  \item \emph{Total for mixing sector:} effectively 1 free parameter
        ($r$), after $\theta_{13}$ is used to fix it.
\end{itemize}
The model predicts two observables ($\theta_{12}$ and $\delta$) from one
input ($\theta_{13}$).  This is genuine predictive power.

The comparison with anarchy (random matrices with no
symmetry)~\cite{deGouvea:2012ac,deGouvea:2003xe} is instructive.
Anarchic models with 7~parameters for 6~neutrino observables also fit
the data, but they predict \emph{no correlations}: the
$(\theta_{23},\,\delta)$ plane is uniformly populated.  $A_4$ with
TM$_1$ predicts a strict one-dimensional curve
(figure~\ref{fig:atm_sum_rule}).  This falsifiable correlation is the
value of the symmetry, not parameter reduction.

\subsection{Robustness of the over-suppression mechanism}
\label{sec:robustness}

The over-suppression failure of $Z_3$ documented in
\S\ref{sec:analytical}--\S\ref{sec:universal} is robust against:
\begin{itemize}
  \item Variations in $\mathcal{O}(1)$ coefficients
        ($10^5$ Monte Carlo samples),
  \item The lightest neutrino mass $m_1\in[0.5,50]\;\mathrm{meV}$,
  \item All Dirac and Majorana CP phases,
  \item All six permutations of the $Z_3$ charges $(2,1,0)$.
\end{itemize}
The median $R\sim 4\times 10^{-11}$ is unchanged across all
variations, confirming its structural origin.

\subsection{Escape routes and their costs}
\label{sec:escape}

One might attempt to rescue $Z_3$ by:
\begin{enumerate}
  \item \textbf{Extreme fine-tuning:} A fraction $\sim 10^{-5}$ of
        parameter space can accidentally produce an adequate solar
        splitting.  This destroys the naturalness of the FN framework.
  \item \textbf{Adding a second flavon:} This effectively extends the
        symmetry beyond $Z_3$, conceding the main point.
  \item \textbf{Abandoning the FN mechanism for neutrinos:} This
        constitutes a sectorial approach, acknowledging that the lepton
        sector requires different dynamics.
\end{enumerate}
Each escape route either abandons the $Z_3$ framework or its
naturalness, reinforcing the conclusion that the lepton sector demands
a structural upgrade.

\subsection{Connection to modular \boldmath $A_4$}
\label{sec:modular}

Recent developments in modular flavour
symmetry~\cite{Feruglio:2017spp,Kobayashi:2018scp} offer a more
economical implementation of~$A_4$.  In modular $A_4$ (the finite
modular group $\Gamma_3\simeq A_4$), Yukawa couplings are modular forms
of a single complex modulus~$\tau$, eliminating flavon fields entirely.
This reduces the parameter count to $\sim 5$, genuinely below the
number of observables~\cite{Feruglio:2017spp}.

The $Z_3$ Froggatt--Nielsen mechanism of
ref.~\cite{Ardakanian:Z3} corresponds precisely to the cusp limit
$\tau\to i\infty$ of this modular framework, where $A_4$ breaks to its
$Z_3^T$ subgroup generated by $T:\tau\to\tau+1$.  In this limit the
three weight-2 modular forms $Y_i(\tau)$ that furnish the $A_4$ triplet
develop a hierarchy $Y_1:Y_2:Y_3\sim 1:|q|^{1/3}:|q|^{2/3}$ with
$q=e^{2\pi i\tau}$, so the expansion parameter of the abelian model is
identified as $\eps\simeq|q|^{1/3}=e^{-2\pi\,\mathrm{Im}\,\tau/3}$.
The value $\eps\approx 0.015$ used in
ref.~\cite{Ardakanian:Z3} corresponds to
$\mathrm{Im}\,\tau\approx 3.2$, well inside the fundamental domain.
King and King~\cite{King:2020weighton} showed that
the modular weights of the fermion fields play the role of
Froggatt--Nielsen charges in this regime (the ``weighton'' mechanism),
while Petcov and Tanimoto~\cite{Petcov:2023cusp} demonstrated explicitly
that the quark mass hierarchies $1:\eps:\eps^2$ and CKM mixing are
reproduced near the cusp with all parameters of natural size.

The seesaw over-suppression and Haar-random PMNS angles identified in
\S\ref{sec:analytical}--\S\ref{sec:universal} can thus be understood as
\emph{cusp-limit artifacts}: deep in the cusp, the $A_4$ triplet
correlations that generate large mixing angles and two independent
neutrino mass scales are exponentially suppressed, leaving only the
residual $Z_3$ structure and its associated pathologies.  A full
unification of the quark and lepton sectors within modular $A_4$ would
require different effective values of~$\tau$ for the two sectors---quarks
near the cusp ($\mathrm{Im}\,\tau\sim 3$) and leptons near a fixed
point ($\mathrm{Im}\,\tau\sim 1$)---which may arise
from distinct moduli in a multi-dimensional compactification
(see ref.~\cite{Ding:2023review} for a comprehensive review).

\subsection{Relation to prior systematic scans}
\label{sec:prior_scans}

Holthausen, Lim, and Lindner~\cite{Holthausen:2012dk} performed a
comprehensive scan of discrete groups up to order~1536, testing their
ability to generate viable mixing patterns.  Their analysis found that
very few groups could produce viable mixing at leading order.  Our work
complements theirs by providing the \emph{physical motivation} for why
$A_4$ is selected: the two-scale problem diagnosed through the $Z_3$
failure (\S\ref{sec:analytical}--\S\ref{sec:universal}) establishes the
structural necessity of a triplet representation, rather than relying
solely on a numerical group-by-group scan.

\subsection{Open questions}
\label{sec:open}

Several aspects deserve further investigation:
\begin{enumerate}
  \item \textbf{Vacuum alignment.}  The VEV alignments
        $\langle\varphi_S\rangle\propto(1,1,1)$ and
        $\langle\chi\rangle\propto(1,0,0)$ are achieved through the
        F-term alignment mechanism with driving
        fields~\cite{Altarelli:2005yp}.  While technically natural, this
        requires additional fields and couplings.  The modular approach
        eliminates this issue.
  \item \textbf{Charged lepton corrections.}  Higher-order operators can
        perturb the diagonal $M_e$, modifying the PMNS matrix at the
        percent level.  These corrections shift $\theta_{12}$ by
        $\mathcal{O}(v_\chi^2/\Lambda^2)$, which could improve or
        worsen the agreement with data.
  \item \textbf{Quark--lepton unification.}  The present analysis treats
        quarks (via $Z_3$~\cite{Ardakanian:Z3}) and leptons (via $A_4$)
        separately.  A complete theory should provide a single UV
        symmetry that breaks to $Z_3$ in the quark sector and $A_4$ in
        the lepton sector.  Candidate unifying groups include $T'$ (the
        double cover of $A_4$, order~24) and $S_4$ (order~24), both of
        which contain $Z_3$ and $A_4$ as
        subgroups~\cite{King:2013eh,Chauhan:2023ppnp}.
\end{enumerate}

%=============================================
\section{Conclusions}
\label{sec:conclusions}
%=============================================

We have systematically investigated the $Z_3$ Froggatt--Nielsen
framework in the neutrino sector---including the type-I seesaw with a
$Z_3$-charged Majorana mass matrix---and shown that $A_4$ is the
minimal discrete symmetry that resolves its structural failures.  The
main findings are:

\begin{enumerate}
  \item When $\MR$ carries the $Z_3$ charge structure dictated by the
        correct Majorana charge algebra---where $\eps$-powers follow
        $(q_i+q_j)\bmod 3$---it contains an unsuppressed off-diagonal
        entry whose dominance in $\MR^{-1}$, combined with the
        hierarchical column texture of $\MD$,
        \textbf{over-suppresses both} $m_1$ \textbf{and} $m_2$ to
        $\mathcal{O}(\eps^3)$ while $m_3$ remains $\mathcal{O}(1)$.

  \item The solar mass splitting is suppressed by \textbf{eight orders
        of magnitude}: $\Delta m^2_{21}/\Delta m^2_{31}\sim
        4\times 10^{-11}$ (median from $10^5$ MC samples) versus $0.030$
        observed.  Only a fraction $\sim 10^{-5}$ of the $\mathcal{O}(1)$
        parameter space approaches the data through accidental
        cancellations.

  \item This failure is \textbf{universal across all six permutations}
        of $Z_3$ charges, following from a number-theoretic property: for
        any permutation of $(2,1,0)$, the pair with charges 1 and 2
        always sums to $0\pmod{3}$.

  \item We show analytically that the generic ratio scales as
        $\Delta m^2_{21}/\Delta m^2_{31}\sim\mathcal{O}(\eps^6)
        \sim 10^{-11}$, in excellent agreement with the Monte Carlo
        median; fewer than $0.01\%$ of parameter-space points
        exceed $\eps^2 \approx 2\times 10^{-4}$, confirmed by the
        numerical scans.

  \item $A_4$ resolves both failures of~$Z_3$.  Its triplet
        representation provides \textbf{two independent mass parameters}
        ($a$ from the singlet flavon $\xi$ and $b$ from the triplet
        flavon $\varphi_S$) that independently control the solar and
        atmospheric mass splittings, solving the two-scale problem.
        The group theory constrains the mixing pattern to TM$_1$, with
        the first column of the PMNS matrix fixed to the TBM form.

  \item The TM$_1$ solar sum rule predicts
        $\sin^2\theta_{12}=0.318$, consistent with NuFit~6.0
        ($1.2\sigma$)~\cite{Esteban:2024nav} and JUNO
        ($1.0\sigma$)~\cite{JUNO:2025}.  The atmospheric sum rule
        provides a parameter-free correlation between $\delta_{\mathrm{CP}}$
        and $\theta_{23}$, predicting $\delta\approx -71^\circ$ at the
        current best-fit point, testable at DUNE~\cite{DUNE:2020} and
        T2HK~\cite{T2HK:2018}.
\end{enumerate}

The coming decade of precision neutrino experiments will subject these
predictions to definitive tests.  JUNO's ultimate measurement of
$\theta_{12}$ to $\sim 0.003$ precision will determine whether the
TM$_1$ solar sum rule survives.  The simultaneous measurement of
$\theta_{23}$ and $\delta$ at DUNE and T2HK will test the atmospheric
sum rule---a strict one-dimensional correlation in the
$(\sin^2\theta_{23},\,\cos\delta)$ plane that distinguishes $A_4$ from
anarchic models.  If confirmed, these correlations would constitute
strong evidence for $A_4$ as a realised symmetry of the lepton sector,
and for the sectorial flavour structure---$Z_3$ for quarks, $A_4$ for
leptons---as a fundamental organising principle of the Yukawa sector.

%=============================================
\acknowledgments
%=============================================

During the preparation of this work the author used Claude (Anthropic)
to assist with numerical computations, Monte Carlo scan implementation,
analytical cross-checks, and manuscript drafting.  After using this
tool, the author reviewed and edited all content, verified all physics
independently, and takes full responsibility for the content of the
published article.

%=============================================

\end{document}